\documentclass[twoside]{article}
\usepackage{qic,epsfig,graphicx,amsmath}

\renewcommand{\thefootnote}{\fnsymbol{footnote}}
\expandafter\let\csname equation*\endcsname\relax
\expandafter\let\csname endequation*\endcsname\relax

\newtheorem{thm}{Proposition}
\newtheorem{lmm}{Lemma}
\newtheorem{cor}{Corollary}

\begin{document}

\setlength{\textheight}{8.0truein}

\runninghead{Detection loophole attacks on
  semi-device-independent quantum and classical protocols}
            {M. Dall'Arno, E. Passaro, R. Gallego,
              M. Paw{\l}owski, A. Ac\'in}

\normalsize\textlineskip
\thispagestyle{empty}
\setcounter{page}{1}

\copyrightheading{0}{0}{2003}{000--000}

\vspace*{0.88truein}

\alphfootnote

\fpage{1}

\centerline{\bf Detection loophole attacks on
  semi-device-independent quantum and classical protocols}

\vspace*{0.37truein}

\centerline{\footnotesize Michele
  Dall'Arno\footnote{cqtmda@nus.edu.sg}}

\vspace*{0.015truein}

\centerline{\footnotesize\it ICFO-Institut de Ciencies
  Fotoniques, Mediterranean Technology Park}
\baselineskip=10pt \centerline{\footnotesize\it
  Castelldefels (Barcelona), 08860, Spain}

\baselineskip=12pt

\centerline{\footnotesize\it Graduate School of Information
  Science, Nagoya University}
\baselineskip=10pt \centerline{\footnotesize\it Chikusa-ku,
  Nagoya, 464-8601, Japan}

\baselineskip=12pt

\centerline{\footnotesize\it Centre for Quantum Technologies, National
  University of Singapore}
\baselineskip=10pt \centerline{\footnotesize\it 3 Science Drive 2,
  117543 Singapore, Singapore}

\vspace*{10pt}

\centerline{\footnotesize Elsa Passaro}

\vspace*{0.015truein}

\centerline{\footnotesize\it ICFO-Institut de Ciencies
  Fotoniques, Mediterranean Technology Park}
\baselineskip=10pt \centerline{\footnotesize\it
  Castelldefels (Barcelona), 08860, Spain}

\vspace*{10pt}

\centerline{\footnotesize Rodrigo Gallego}

\vspace*{0.015truein}

\centerline{\footnotesize\it ICFO-Institut de Ciencies
  Fotoniques, Mediterranean Technology Park}
\baselineskip=10pt \centerline{\footnotesize\it
  Castelldefels (Barcelona), 08860, Spain}

\baselineskip=12pt \centerline{\footnotesize\it Dahlem
  Center for Complex Quantum Systems, Freie Universit\"at
  Berlin,}
\baselineskip=10pt \centerline{\footnotesize\it 14195
  Berlin, Germany}

\vspace*{10pt}

\centerline{\footnotesize Marcin Paw{\l}owski}

\vspace*{0.015truein}

\centerline{\footnotesize\it Institute of Theoretical
  Physics and Astrophysics, University of Gdansk,}
\baselineskip=10pt \centerline{\footnotesize\it 80-952
  Gdansk, Poland}

\baselineskip=12pt \centerline{\footnotesize\it Department
  of Mathematics, University of Bristol}
\baselineskip=10pt \centerline{\footnotesize\it Bristol BS8
  1TW, United Kingdom}

\vspace*{10pt}

\centerline{\footnotesize 
Antonio Ac\'in}

\vspace*{0.015truein}

\centerline{\footnotesize\it ICFO-Institut de Ciencies
  Fotoniques, Mediterranean Technology Park}
\baselineskip=10pt \centerline{\footnotesize\it
  Castelldefels (Barcelona), 08860, Spain}

\baselineskip=12pt \centerline{\footnotesize\it
  ICREA-Institucio Catalana de Recerca i Estudis
  Avan\c{c}ats}
\baselineskip=10pt \centerline{\footnotesize\it Lluis
  Companys 23, 08010 Barcelona, Spain}

\vspace*{0.225truein}

\publisher{(received date)}{(revised date)}

\vspace*{0.21truein}

\abstracts{
  Semi-device-independent quantum protocols
  realize information tasks -- e.g. secure key distribution,
  random access coding, and randomness generation -- in a
  scenario where no assumption on the internal working of
  the devices used in the protocol is made, except their
  dimension. These protocols offer two main advantages:
  first, their implementation is often less demanding than
  fully-device-independent protocols. Second, they are more
  secure than their device-dependent counterparts. Their
  classical analogous is represented by random access codes,
  which provide a general framework for describing one-sided
  classical communication tasks. We discuss conditions under
  which detection inefficiencies can be exploited by a
  malicious provider to fake the performance of
  semi-device-independent quantum and classical protocols --
  and how to prevent it.}
{}{}
\vspace*{10pt}

\vspace*{3pt}
\communicate{}

\vspace*{1pt}\textlineskip

\section{Introduction}      
\setcounter{footnote}{0}
\renewcommand{\thefootnote}{\alph{footnote}}
\noindent
In the last decades, the distinguishing properties of
quantum theory have been exploited to accomplish tasks which
are unfeasible in classical theory~\cite{NC00}. For example,
protocols were proposed for secure quantum key distribution
(QKD)~\cite{BB84, E91}, quantum
teleportation~\cite{BBCJPW93, BPMEWZ97}, and quantum
randomness generation (QRG)~\cite{ROT94, FSSBM07}. 
With no exceptions, the first protocols to
be proposed were device dependent, namely their success
critically relies on the agreement between the description
of the setup and its implementation. Since this hypothesis
is never exactly fulfilled, in experimental implementations
a plethora of related problems arises. In fact, the
unavoidable mismatch between theoretical requirements of
security proofs and experimental implementations has
recently been exploited to hack quantum cryptography
systems~\cite{LWWESM10, GLLSKM11, XQL10, JWLWEMML11}.

Subsequently, fully-device-independent protocols were
proposed, in a scenario where devices are completely
uncharacterized. Celebrated examples of device-independent
protocols include schemes for randomness
generation~\cite{PAMBMMOHLMM10, Col07, CK11} or secure
quantum key distribution~\cite{ABGMPS07,
  PABGSS09,BHK05}. While first security proofs of these
protocols were partial and required some extra assumptions,
such as, for instance, to assume that devices were
memoryless, general security proofs have recently been
obtained in~\cite{VV12, RUV12}. Unfortunately, these
security proofs turn out to be quite demanding in terms of
noise (in fact, the proof in~\cite{RUV12} does not tolerate
any noise). Moreover, from an implementation point of view,
device-independent protocols are quite challenging, as they
require a detection-loophole-free Bell inequality violation
between two distant parties.

Recently, intermediate solutions between device-dependent
and fully device-independent protocols have been proposed.
For example, in the so-called measurement-device-independent
(MDI) QKD~\cite{BP12,LCQ12}, the measurement device is left
uncharacterized. While offering a weaker form of security
than their device-independent counterparts, the
implementation of MDI protocols is much easier as they can
tolerate higher losses and their performance is comparable
with that of standard schemes. Security of
measurement-device-independent QKD protocols under detection
loophole hacks is discussed in Ref.~\cite{LCQ12}, and
experimental implementations are reported in
Refs.~\cite{TLXQQL14, LCWLSWCYLLMPFPZP13}.

Another example are the so-called semi-device-independent
quantum protocols, where only the dimension of the exchanged
system is assumed while the devices are uncharacterized. At
the price of upper bounding the dimension of the system,
secure QKD is possible~\cite{PB11} in a measure and prepare
scheme, and semi-device-independent protocols for QRG are
also known~\cite{HPZGZ11}. The importance of semi-device
independent protocols relies on the fact that they are
secure against a wide class of attacks, where the
eavesdropper can alter the characteristics of systems and
measurements, but can not change their dimensionality. For
example, in the attacks proposed in Refs.~\cite{LWWESM10,
  GLLSKM11, JWLWEMML11}, by flashing bright light into the
devices, an eavesdropper gets control of detection
efficiency for each measurement, but there is no way for him
to make the devices encode the information in a system of a
higher dimensionality.
 
Semi-device-independent protocols are based on the quantum
certification provided by dimension witnesses for a fixed
dimension~\cite{GBHA10}. While the recently developed
formalism of dimension witness~\cite{GBHA10, HGMBAT11,
  ABCB11} allows to lower bound the dimension of a system in
a device-independent fashion, a device-independent upper
bound is prohibited by fundamental arguments. Indeed, for
any given system, no experiment is guaranteed to exploit all
the available degrees of freedom. Then, only some a priori
knowledge of the underlying physical model can provide a
(device dependent) upper bound on the dimension of the
system. On the other hand, it is worth to stress that in
semi-device-independent protocols no additional assumption
on the model needs to be done. For example, for information
encoded in the polarization of a photon, only the
unsurprising assumption that the system is two dimensional
is required, while no assumption is required on what the
polarization directions are.

The classical analogous of semi-device-independent quantum
protocols -- namely, the case in which the exchanged system
is classical -- is known as the problem of random access
codes (RACs)~\cite{ALMO08}. In the context of RACs, the aim
of two distant parties is to optimally perform some
one-sided communication task under a constraint on the
amount of classical information exchanged.

Despite their security, real world implementations of
semi-device-independent (quantum or classical) protocols are
subject to detection loophole (DL) attacks~\footnote{The
  problem of DL attack on semi-device-independent protocols
  shares analogies with that of the robustness to loss of
  device-independent dimension witness (DIDWs)
  \cite{GBHA10,HGMBAT11,ABCB11,BQB13}, addressed in
  \cite{DPGA12}. Nevertheless, while in the former the task
  is to exploit non-ideal detection efficiencies to fake the
  result of a protocol producing an input/output statistics
  which would be forbidden in the absence of DL, in the
  latter the task is to devise conditions under which
  dimension witnessing is indeed possible even in the
  presence of loss.}  -- as happens for any
fully-device-independent protocol. In this attack, a
malicious provider exploits non-ideal detection efficiencies
to skew the statistics of the experiment and ultimately
faking its result. The aim of this work is to provide
conditions under which DL attacks are harmless in faking the
result of a semi-device-independent (quantum or classical)
protocol.

The paper is organized as follows. In Sec.~\ref{sect:dl} we
introduce DL attacks and present our main results. In
Sec.~\ref{sect:quantprot} we derive conditions under which
DL attacks on semi-device-independent quantum protocol are
harmless, in the general framework where only the statistics
of the protocol is taken into account. In
Sec.~\ref{sect:classprot} we address the problem of the
certification of semi-device-independent classical
protocols, in the framework of RACs. Finally we summarize
our results and delineate some further developments in
Sec.~\ref{sect:conclusion}.

\section{Detection loophole attack}
\label{sect:dl}
\noindent
We start by presenting the general structure of the
semi-device-independent (quantum and classical) protocols
considered in this work. The existing quantum protocols
for QKD~\cite{PB11} and QRG~\cite{HPZGZ11}, as well as
classical RACs, are examples of this structure. We
consider protocols in which two distant parties, Alice and
Bob, have access to uncorrelated random number
generators\footnote{Notice that the assumption of
  uncorrelation is fulfilled by a broad class of
  protocols. Indeed, in any semi-device-independent setup
  one necessarily has to assume that the devices are
  shielded -- namely they can not communicate except
  through message $A$. Then to have shared (classical or
  quantum) randomness one is forced to introduce a trusted
  third party random generator, or to allow for infinite
  local memory on each device storing previously
  distributed randomness.}. 
For each round, we
denote by $j$ ($i$) the random variable generated by
Alice's (Bob's) generator and with $q_j$ ($p_i$) its
probability distribution. As said, these probability
distributions are independent. Random variables $j$ and
$i$ represent the strategy that Alice and Bob apply,
respectively. This scheme is depicted in
Fig.~\ref{fig:setup}.

\begin{figure}[htbp]
  \centerline{\epsfig{file=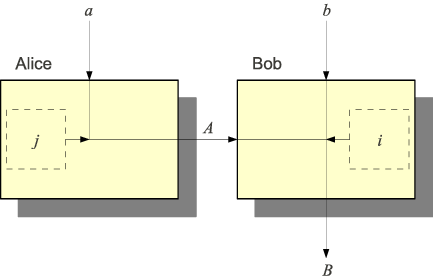, width=8.2cm}}
  \vspace*{13pt} \fcaption{\label{fig:setup} Scheme of a
    generic semi-device-independent (quantum or classical)
    protocol. Two distant parties Alice and Bob are provided
    a black box each (bold-line boxes in the
    Figure). Alice's and Bob's boxes receive classical input
    $a$ and $b$ respectively. Each box is allowed to use a
    classical random generator (dashed-line boxes), which
    outcome -- $j$ for Alice's box and $i$ for Bob's -- is
    not accessible to the parties but can influence the
    outcome of the box. Alice's box sends a (quantum or
    classical) message $A$ to Bob, that finally outputs
    classical message $B$.}
\end{figure}

In each run, Alice and Bob get classical inputs $a$ and $b$
respectively. Alice sends a message $A$ -- which may be
classical or quantum -- to Bob, who then returns a classical
value $B$. Finally they collect the statistics of several
runs (the asymptotic case is always considered), obtaining
the conditional probability distribution $P(B|a,b)$ of
outcome $B$ given inputs $a$ and $b$, namely
\begin{align}\label{eq:stat}
  P(B|a,b) := \sum_{i,j,A} p_i q_j P_i(B|A,b) P_j(A|a).
\end{align}
It is important to stress that -- as Eq.~\eqref{eq:stat}
clearly shows -- access is granted only to the inputs $a,b$
and the output $B$, while no knowledge of the internal
behavior of the black boxes (including the random variables
$i,j$) and of the message $A$ is provided. The goal is to
exploit the correlations between the two parties,
encapsulated by $P(B|a,b)$ to solve an information task,
e.g. to distribute a secure key or generate random numbers.

When studying DL attacks, we assume that for each round of
the experiment Alice or Bob can claim that their ``detector
did not click'', and in this case this round of the
experiment is discarded from the statistics. In general,
Alice’s box can decide whether to send state $A$ to Bob
after receiving her input $a$ and random variable $j$, while
Bob's box after receiving his input $b$, the message $A$ and
random variable $i$. Thus, the detection efficiencies,
i.e. the probabilities that the detector clicks, are denoted
with $\eta_j(a)$ for Alice and $\eta_i(A,b)$ for Bob. Notice
that these probabilities cannot be estimated as they depend
on the variables $i$ and $j$ internal to the devices. The
conditional probability distribution of outcome $B$ given
inputs $a$ and $b$ in the presence of a DL attack is given
by
\begin{align}\label{eq:pdel1}
  P_{DL}(B|a,b) := \frac {\sum_{i,j,A} p_i q_j \eta_i(A,b)
    \eta_j(a) P_i(B|A,b) P_j(A|a)} {\sum_{i,j,A} p_i q_j
    \eta_i(A,b) \eta_j(a) P_j(A|a)}.
\end{align}
We use the suffix DL whenever a distribution is obtained
resorting to DL attack. Without loss of generality, we are
assuming that for every input $a,b$ there is a non-zero
probability of click, namely denominator in
Eq.~\eqref{eq:pdel1} is strictly larger than $0$ for any $a$
and $b$. Indeed, if this is not the case, namely if there
exist $a$ and $b$ such that the corresponding probability of
click is zero, one can discard their occurrences from the
statistics without affecting the success of the protocol.

Notice that whether Alice uses DL is not
relevant\footnote{This may be no more true if other
  constraints are introduced, since in this case the sets of
  distributions $P(A|a)$ and $P_{DL}(A|a)$ attainable by
  Alice can be different. For example, suppose that Alice is
  computationally constrained to prepare message $A$ in time
  polynomial in the size of $a$. On the one hand, without
  resorting to DL it is impossible to obtain the
  distribution $P(A|a) = \delta_{A,f(a)}$, with $f(a)$ some
  NP-hard function (as long as we assume that $P \neq
  NP$). On the other hand, exploiting DL Alice can randomly
  choose $A$ and check in polynomial time whether $A =
  f(a)$, clicking only in this case.}, since any settings
she can prepare with DL can also clearly be achieved without
resorting to it, so Eq.~\eqref{eq:pdel1} can be simplified
as
\begin{align}\label{eq:pdel2}
  P_{DL}(B|a,b) = \frac {\sum_{i,A} p_i \eta_i(A,b)
    P_i(B|A,b) P(A|a)} {\sum_{i,A} p_i \eta_i(A,b)
    P(A|a)},
\end{align}
where $P(A|a) := \sum_j q_j \eta_j(a) P_j(A|a)/\sum_j q_j
\eta_j(a)$.

Independently of the task to be realized, all the known
examples of semi-device-independent quantum protocols are
based on the quantum certification provided by dimension
witnesses~\cite{GBHA10} or, in other words, on the fact
that, for a fixed dimension of the exchanged system $A$,
there are quantum distributions $P(B|a,b)$ that cannot be
attained when system $A$ is classical - a system is
classical if the states in which it can be prepared are
pairwise commuting. This quantum certification plays here
the same role as Bell violations for fully
device-independent protocols. Our purpose is then to
understand how a DL attack can mimic correlations that are
intrinsically quantum exploiting the losses in the
implementation. That is, rather than analyzing the effect of
losses for a given quantum protocol, we study situations in
which the observed correlations are useless for any quantum
protocol. This is analogous to what is done when studying
the detection loophole for Bell inequalities.

On the other hand, for classical protocols such a general
approach is obviously not possible. However, when
addressing the problem of classical RACS, one is usually
interested in maximizing some figure of merit related to
the particular communication task, such as the worst case
or average probability of correct detection. In this work
we will focus on the former - being the latter related by
Yao's principle~\cite{ALMO08,Yao77} -- and we will devise
conditions under which it can not be improved resorting to
DL attack.

\section{Certification of semi-device-independent quantum protocols}
\label{sect:quantprot}
\noindent
In this Section, we focus on semi-device-independent quantum
protocols, namely where the exchanged system $A$ is
quantum. As mentioned, the success of
semi-device-independent quantum protocols depends on the
generated statistics. Usually, for a given protocol, a large
enough value of a particular function of such statistics
ensures the success of the protocol. For instance, the
protocol in~\cite{PB11} is secure only when it is assumed
that the dimension of the measured systems is two and a
large value of a dimension witness is observed. Yet, in
general, a necessary condition for the successful
performance of any protocol is the ability to discriminate
whether the source is intrinsically quantum or it can be
described as a classical distribution, building only on the
knowledge of the conditional probability distribution
$P(B|a,b)$.  That is, it is necessary to certify that the
observed correlations cannot be explained classically and,
therefore, are potentially useful for quantum protocols
without classical analogue. The advantage of this approach
is that it allows one to evaluate necessary conditions for
security irrespectively of the particular protocol
considered. Indeed, finding a DL attack able to fake an
intrinsically quantum distribution by exploiting detection
inefficiencies makes the observed correlations useless for
any protocol. In this Section we provide conditions under
which DL attack can by no means recast a classical
$P(B|a,b)$ into an intrinsically quantum $P_{DL}(B|a,b)$
thus faking the result of the protocol.

We say that a conditional probability distribution
$P(B|a,b)$ of outcome $B$ given inputs $a$ on Alice's side
and $b$ on Bob's side admits a classical (quantum)
$d$-dimensional model if it can be written as
\begin{align*}
  P(B|a,b) = \sum_{A,i} p_i P_i(B|A,b) P(A|a),
\end{align*}
where
\begin{align}\label{eq:norm}
  \sum_A P(A|a) = 1, \; \forall a, \quad \sum_B
  P_i(B|A,b)=1, \; \forall A,b,i,\\ P(A|a) \ge 0 \;
  \forall A,a, \quad P_i(B|A,b) \ge 0 \; \forall
  B,A,b,i\nonumber
\end{align}
for some probability $p_i$ and where $A$ is a classical
(quantum) $d$-dimensional system. Given some correlations
with losses, we say that DL attacks are harmless whenever
there is no classical attack faking the correlations.

The probability of click on Bob's side given he received
message $A$ from Alice and input $b$ is given by
\begin{align*}
  Q(B \neq NC|A,b) := \sum_i p_i \eta_i(A,b),
\end{align*}
where we denoted with $NC$ the no-click event. The
following Lemma shows that whenever
\begin{align}\label{messindep}
  Q(B{\neq}NC|A,b)= Q(B{\neq}NC|b).
\end{align}
then, DL attacks are harmless.

\vspace*{12pt}
\noindent
\begin{lmm}\label{thm:nodl1}
  If $Q(B{\neq}NC|A,b)=Q(B{\neq}NC|b)$ for any $A, b$,
  then if $P_{DL}(B|a,b)$ does not admit a $d$-dimensional
  classical (quantum) model then also $P(B|a,b)$ does not
  admit a $d$-dimensional classical (quantum) model.
\end{lmm}

\vspace*{12pt}
\noindent
{\bf Proof:}
  First we show that under the hypothesis
  $Q(B{\neq}NC|A,b) = Q(B{\neq}NC|b)$, if $P(B|a,b)$
  admits a $d$-dimensional classical (quantum) model then
  also $P_{DL}(B|a,b)$ admits a $d$-dimensional classical
  (quantum) model.

  Let then us assume that $P(B|A,b)$ admits a classical
  (quantum) model, namely it can be written as
  \begin{align*}
    P(B|a,b) = \sum_{A,i} p_i P_i(B|A,b) P(A|a),
  \end{align*}
  with $P_i(B|A,b), P(A|a)$ satisfying Eq. \eqref{eq:norm}
  and for some probability $p_i$. Then by definition
  \begin{align*}
    P_{DL}(B|a,b) = \frac {\sum_{i,A} p_i \eta_i(A,b)
      P_i(B|A,b) P(A|a)} {\sum_{i,A} p_i \eta_i(A,b)
      P(A|a)}.
  \end{align*}

  Upon introducing the hypothesis $Q(B{\neq}NC|A,b) =
  Q(B{\neq}NC|b)$ one has
  \begin{align*}
    P_{DL}(B|a,b) = \frac {\sum_{i,A} p_i \eta_i(A,b)
      P_i(B|A,b) P(A|a)} {Q(B{\neq}NC|b)}.
  \end{align*}
  Upon setting
  \begin{align*}
    P_{DL}(B|A,b) = \frac {\sum_i p_i \eta_i(A,b)
      P_i(B|A,b)} {Q(B {\neq}NC|b)},
  \end{align*}
  one clearly has $\sum_B P_{DL}(B|A,b) = 1$ and
  $P_{DL}(B|A,b) \ge 0$ for any $B,A,b$. Then
  $P_{DL}(B|a,b)$ admits the $d$-dimensional classical
  (quantum) model
  \begin{align*}
    P_{DL}(B|a,b) = \sum_A P_{DL}(B|A,b) P(A|a).
  \end{align*}
  Then, whenever $P_{DL}(B|a,b)$ does not admit a
  $d$-dimensional classical (quantum) model, also
  $P(B|a,b)$ does not.$\square$

At this point it is convenient to discuss condition
\eqref{messindep} in relation with the so-called {\em fair
  sampling assumption}. The latter states that the set of
events in which the detectors clicked is a randomly chosen
sample from the total set of events that one would have
obtained with perfect detectors. That is,
\begin{align}\label{fairsampling}
  P_{DL}(B|a,b)=P(B|a,b).
\end{align}
One can clearly see by using Eq.~\eqref{eq:pdel2}, that
\eqref{messindep} does not necessarily imply the fair
sampling assumption. Indeed, in order to fulfill the fair
sampling assumption for every choice of
$\{p_i,P(A|a),P_i(B|A,b)\}$, one needs that
$\eta_{i}(A,b)=\eta(b)$. On the other hand, in order to
fulfill \eqref{messindep} for every function $p_i$
suffices that $\eta_{i}(A,b)=\eta_i(b)$. In this sense,
Lemma \ref{thm:nodl1} generalizes the fair sampling
assumption, providing strictly weaker hypothesis under
which DL-attacks are harmless.

Nonetheless, the fair sampling assumption and our slightly
more general condition~\eqref{messindep} have in common
that they refer to properties of the internal working of
the devices. In particular, condition
$\eta_{i}(A,b)=\eta_i(b)$ -- or the more constraining fair
sampling assumption -- cannot be verified solely from the
statistics, since the message $A$ sent by Alice is not
directly accessible to the parties.

Next Proposition provides a much stronger condition for DL
attacks to be harmless, as it is stated only in terms of the
probability $Q(B{\neq}NC|a,b)$ of click given inputs $a$ on
Alice's side and $b$ on Bob's side, namely
\begin{align}
  Q(B{\neq}NC|a,b) := \sum_{i,A} p_i \eta_i(A,b) P(A|a).
\end{align}
Notice that this probability is accessible to the parties,
being a function of the inputs $a,b$ which are in turn
accessible.  The following proposition shows that whenever
statistics of bidimensional systems fulfill
\begin{align}\label{inputindep}
  Q(B{\neq}NC|a,b) = Q(B{\neq}NC|b)
\end{align}
DL-attacks are harmless.

\vspace*{12pt}
\noindent
\begin{thm}\label{thm:nodl2}
  If $Q(B{\neq}NC|a,b) = Q(B{\neq}NC|b)$ for any $a, b$,
  then if $P_{DL}(B|a,b)$ does not admit a $2$-dimensional
  classical model then also $P(B|a,b)$ does not admit a
  $2$-dimensional classical model.
\end{thm}

\vspace*{12pt}
\noindent
{\bf Proof:}
  By hypothesis, for any input $a_0, a_1$ on Alice's side
  one has
  \begin{align*}
    \sum_A Q(B{\neq}NC|A,b) \left[ P(A|a{=}a_0) -
      P(A|a{=}a_1) \right] = 0,
  \end{align*}
  where the sum is over $A = 0,1$.

  Rearranging explicitly the terms in previous Equation
  and using the fact $P(A=1|a) = 1 - P(A=0|a)$ for any
  $a$, one obtains that either
  \begin{align*}
    P(A{=}0|a{=}a_0) = P(A{=}0|a{=}a_1),
  \end{align*}
  for any $a_0, a_1$, namely the message $A$ sent by Alice
  is independent on her input $a$, or
  \begin{align*}
    Q(B{\neq}NC|A{=}0,b) = Q(B{\neq}NC|A{=}1,b),
  \end{align*}
  for any $b$, namely the detection probability on Bob's
  side is independent on the message $A$ received from
  Alice.

  In the former case $P(B|a,b)$ clearly admits a classical
  local model, namely one in which no message is sent from
  Alice to Bob, and the same holds true for
  $P_{DL}(B|a,b)$ due to Eq.~\eqref{eq:pdel2}. In the
  latter case the hypothesis of Lemma~\ref{thm:nodl1} is
  satisfied, and thus the statement is proven.$\square$

As said, contrary to Lemma~\ref{thm:nodl1}, this result is
much stronger, as it is proven under an assumption that can
be verified only from the observed statistics. The price to
pay is that it only holds for systems of dimension two.
Condition~\eqref{inputindep} is, therefore, highly
inequivalent to~\eqref{messindep} or the fair sampling
assumption. In fact, the attack presented
in~\cite{LWWESM10,GLLSKM11} fulfills
condition~\eqref{inputindep}, however clearly violates the
fair sampling assumption (but also violates the assumption
on the dimension).

\section{Certification of semi-device-independent classical protocols}
\label{sect:classprot}
\noindent
In this Section, we focus on semi-device-independent
classical protocols, namely where the exchanged system $A$
is classical. We devise functions of the input/output
statistics that can not be altered by DL attacks. Thus,
any certification for semi-device-independent classical
protocols building only on the value of these functions
will be immune to DL attacks. Again, the main advantage is
that, as above, these functions can be verified only from
the observed statistics.

A semi-device-independent classical protocol can be viewed
as a random access code~\cite{ANTV02,HINRY06} (RAC), and
in the following it will be convenient to work in the
framework of RACs. In this framework, the aim of the two
distant parties Alice and Bob is to optimally perform some
communication task by means of one-sided communication of
classical information. RACs are usually denoted with the
notation $n \to m$. Here $n$ is the number of input bits
of Alice, namely the dimension of input $a$ is $\dim(a) =
2^n$, while $m$ is the number of bits sent by Alice,
namely the dimension of message $A$ is $\dim(A) = 2^m$
(see Fig.~\ref{fig:setup}).

In this scenario, the relevant figures of merit usually
considered are the worst case or the average success
probability to have that $B = f(a,b)$ for a specific
Boolean function $f(a,b) \in \{0,1\}$. Here we will focus
on the former, being the latter related through Yao's
principle~\cite{Yao77}. The worst case probability of
success $P^{wc}$ is defined as
\begin{align*}
  P^{wc} := \min_{a,b} P(B{=}f(a,b)|a,b).
\end{align*}
The probability that $B=f(a,b)$ with the DL attack is given
by
\begin{align}\label{eq:psuccdel}
  P_{DL}(B{=}f(a,b)|a,b) = \frac{\sum_{i,A} w_i(A,a,b)
    P_i(B{=}f(a,b)|A,b)}{\sum_{i,A} w_i(A,a,b)},
\end{align}
where $w_i(A,a,b) = p_i \eta_i(A,b) P(A|a)$ and the worst
case probability that $B = f(a,b)$ is given by
\begin{align*}
  P_{DL}^{wc} := \min_{a,b} P_{DL}(B{=}f(a,b)|a,b).
\end{align*}

The following Proposition provides conditions under which
the worst case success probability of a RAC can not be
increased resorting to DL exploit. When these hypotheses
are satisfied, a protocol relying on the worst case
success probability may not be affected by DL attack.

\vspace*{12pt}
\noindent
\begin{thm}\label{thm:nodl3}
  Given a RAC, if the worst case success probability
  without resorting to DL attack is $P^{wc}{=}1/2$, then
  the worst case probability of success resorting to DL
  attack is $P_{DL}^{wc}{=}1/2$.
\end{thm}

\vspace*{12pt}
\noindent
{\bf Proof:}
  The proof proceeds by absurd assuming $P^{wc} = 1/2$ and
  $P_{DL}^{wc} > 1/2$.

  Equation~\eqref{eq:psuccdel} is the weighted sum over
  indices $i$ and $A$ of the numbers $P_i(B=f(a,b)|A,b)$
  with weights $w_i(A,a,b)/\sum_{i,A} w_i(A,a,b)$ and
  therefore is upper bounded by
  \begin{align*}
    P_{DL}(B{=}f(a,b)|a,b)\leq \max_{A,i} \{
    P_i(B{=}f(a,b)|A,b)\},
  \end{align*}
  and one has
  \begin{align*}
    P_{DL}^{wc} \leq \min_{a,b} \max_{A,i} \{
    P_i(B{=}f(a,b)|A,b)\}.
  \end{align*}

  Since we are assuming $P_{DL}^{wc} > 1/2$ there exists a
  strategy $i_0$ of Bob and a message $A_0$ of Alice such
  that for all $a,b$ one has $P_{i_0}(B=f(a,b)|A_0,b) >
  1/2$. Then Bob can exploit a new strategy where he
  applies strategy $i_0$ whenever he gets $A_0$ and
  returns a random number otherwise, for which the
  probability $\tilde{P} (B{=}f(a,b)|a,b)$ of $B{=}f(a,b)$
  given inputs $a$ and $b$ is given by
  \begin{align*}
    \tilde{P}(B{=}f(a,b)|a,b) = \left[
      P_{i_0}(B{=}f(a,b)|A_0,b) - \frac12 \right] P(A_0|a)
    +\frac12.
  \end{align*}

  This new strategy does not resort to DL and since
  $P_{i_0}(B=f(a,b)|A_0,b) > 1/2$ it has the worst case
  success probability greater than
  \begin{align*}
    \tilde{P}^{wc} = \min_{a,b}
    \{\tilde{P}(B{=}f(a,b)|a,b)\} > \frac{1}{2}
  \end{align*}
  which contradicts the assumptions.$\square$

The following Corollary shows that for any $n \to 1$ RAC,
DL attacks are harmless.
\vspace*{12pt}
\noindent
\begin{cor}\label{thm:nodl4}
  For any $n \to 1$ RAC the worst case success probability
  resorting to DL attack is $P_{DL}^{wc}{=}1/2$.
\end{cor}

\vspace*{12pt}
\noindent
{\bf Proof:}
  In \cite{ANTV99} it was shown that for any $n \to 1$ RAC
  the hypothesis of Proposition~\ref{thm:nodl3} are
  fulfilled, namely $P^{wc} = 1/2$, so the statement
  follows.$\square$

One may ask whether it is possible to relax the hypothesis
of Proposition~\ref{thm:nodl3}. We provide here an example
of RAC with worst case success probability larger than
$1/2$, and show that this probability can be increased using
DL attack. Consider the $3 \to \log6$ RAC. Alice is given
three independent bits $a_0,a_1,a_2$, namely $a = a_0
\otimes a_1 \otimes a_2$, and she can send to Bob a
$6$-dimensional message or, equivalently, one bit $A_0$ and
one trit $A_1$, namely $A = A_0 \otimes A_1$. Bob's input is
the trit $b = 0,1,2$ and the function to be computed is
$f(a,b) = a_b$. Here we show that the worst case success
probability $P^{wc}$ without resorting to DL of $3 \to \log
6$ RAC is $P^{wc} < 0.981$, while there exists a DL attack
such that the worst case success probability is $P_{DL}^{wc}
= 1$.

First, we prove that for the $3 \to \log 6$ RAC one has
$P^{wc} < 0.981$. An explicit upper bound for the worst
case quantum success probability -- which is clearly at
least as large as the classical one $P^{wc}$ -- was
derived in~\cite{Nay99} in the context of quantum finite
automata, namely
\begin{align*}
  (1-h(P^{wc})) n \leq m,
\end{align*}
where $h(.)$ is the Shannon binary entropy
function. Setting $n=3$ and $m=\log 6$ we get $P^{wc} <
0.981$.

Now we provide a protocol using DL which achieves the
worst case success probability $P_{DL}^{wc} = 1$. The idea
is to use part of the communicated message to distribute
randomness. Alice can choose the trit $A_1$ at random and
encode $A_0=a_{A_1}$, in other words she sends one of her
bits randomly to Bob but also sends him information
regarding which bit it is. If $b=A_1$ then Bob returns
$B=A_0$ which is equal to $a_b$. If $b\neq A_1$ his
detector does not click. The detection efficiency of Bob's
device with this protocol is given by
$\eta_i(A,b)=\delta_{b,A_1}$, and the worst case success
probability is given by $P_{DL}^{wc}=1$.

\section{Conclusion}
\label{sect:conclusion}
\noindent
In this work we addressed the problem of how non-ideal
detection efficiencies can be exploited to fake the result
of semi-device-independent quantum and classical protocols
through DL attacks.

For quantum protocols, we discussed general conditions under
which DL attacks are harmless in terms of the detection
probability. First, we proved in Lemma~\ref{thm:nodl1} that
a sufficient condition is that the measurements can be
modeled as ideal ones affected by non ideal detection
efficiency - in terms of positive operator-valued measures
(POVMs)~\cite{NC00}, this means that the element
corresponding to no-click event is proportional to the
identity. Perhaps surprisingly, the ``quantumness'' of the
generated statistics does not depend on the detection
efficiency, as far as this is strictly larger than zero,
thus making our certification strategies extremely
robust. Unfortunately, this condition is not practical, as
it requires some additional assumption on the model
underlying the measurements. Then, we derived in
Proposition~\ref{thm:nodl2} a sufficient condition which is
practical, namely it can be verified only from the knowledge
of the input/output statistics. Also in this case the
certification is extremely robust, as it can be obtained for
any non-null value of the detection efficiency. Notice that,
after the submission of this work, an analogous result for
device independent dimension witnessing with uncorrelated
devices was reported in Ref.~\cite{BQB13}. There, indeed, it
was shown that a particular dimension witness can certify
the quantumness of a qubit for any non-null value of the
detection efficiency.

For classical protocols, we provided conditions under which
DL attacks can not increase the worst case success
probability of a RAC. Our main results can be used as a
guideline to devise quantum and classical protocols
resistant to DL attacks, being thus of relevance for
applications such as QKD, QRG, and RAC. Our approach for
quantum protocols is very general, providing necessary
conditions for the quantum certification of the devices. A
natural follow-up question is to understand how DL attacks
apply to specific examples of semi-device-independent
quantum and classical protocols.

Some of the presented results -- namely
Proposition~\ref{thm:nodl2} and Corollary~\ref{thm:nodl4}
- hold in the hypothesis that the message $A$ sent by
Alice is $2$-dimensional. For the classical case, we
showed through the example $3 \to \log6$ RAC that this
assumption can not be relaxed trivially. Thus, it remains
as an open problem how to devise more general conditions
under which DL attacks are harmless.

\nonumsection{Acknowledgements}
\noindent
We are grateful to Nicolas Brunner for very useful
discussions and suggestions. M. D. thanks Anne Gstottner and
the Human Resources staff at ICFO for their invaluable
support. This work was funded by the Spanish project
FIS2010-14830 and Generalitat de Catalunya, UK EPSRC, the
European PERCENT ERC Starting Grant and Q-Essence Project,
the JSPS (Japan Society for the Promotion of Science)
Grant-in-Aid for JSPS Fellows No. 24-0219, the Excellence
Initiative of the German Federal and State Governments
(Grant ZUK 43), NCN grant no. 2013/08/M/ST2/00626, the
Ministry of Education and the Ministry of Manpower
(Singapore).

\nonumsection{References}

\end{document}